\documentclass[conference]{IEEEtran}

\usepackage{amsmath, amssymb, graphicx, dsfont, setspace,url, amsfonts, amsthm, multirow,mathabx}
\usepackage{ amssymb, graphicx, dsfont, setspace, url, subfigure}

\DeclareMathOperator{\supr}{sup}
\newtheoremstyle{Amin}
  {3pt}
  {3pt}
  {}
  {}
  {\bfseries}
  {:}
  {.5em}
  {}

\theoremstyle{Amin}

\newtheorem{remark}{Remark}
\newtheorem{defin}{Definition}
\newtheorem{prop}{Proposition}

\newtheorem{theorem}{Theorem}

\usepackage[T1]{fontenc}
\usepackage[latin9]{inputenc}

\usepackage{setspace}

\begin{document}
\setcounter{page}{1}
\title{Semi-Quantitative Group Testing}

\author{\IEEEauthorblockN{Amin Emad and  Olgica Milenkovic}
\authorblockA{University of Illinois, Urbana-Champaign, IL}
E-mail: \{emad2,milenkov\}@illinois.edu
}

\maketitle

\begin{abstract} We consider a novel group testing procedure, termed \emph{semi-quantitative group testing}, motivated by a class of problems arising in genome sequence processing. Semi-quantitative group testing (SQGT) is a non-binary pooling scheme that may be viewed as a combination of an adder model followed by a quantizer. For the new testing scheme we define the capacity and evaluate the capacity for
some special choices of parameters using information theoretic methods. We also define a new class of disjunct codes suitable for SQGT, termed \emph{SQ-disjunct} codes. 
We also provide both explicit and probabilistic code construction methods for SQGT with simple decoding algorithms.
\end{abstract}


\section {Introduction}
Group testing (GT) is a pooling scheme for identifying a number of subjects with some particular characteristic -- called ``positives'' -- among a large pool of subjects. The idea behind GT is that if the number of positives is much smaller than the number of subjects, one can reduce the number of experiments by testing adequately chosen groups of subjects rather than testing each subject individually. In its full generality, GT may be viewed as the problem of inferring the state of a system from a superposition of the state vectors of a subset of the system's elements. As such, it has found many applications in communication theory, signal processing, computer science, mathematics, biology, etc. (for example see~\cite{DH00}-\cite{FDH95}). 

Although many models have been considered for combinatorial GT,  two main models include the original model considered by Dorfman~\cite{D43} (henceforth, conventional GT) and the adder model (also known as the adder channel or quantitative GT)~\cite{DH00}. In the former case, the result of a test is an indicator determining if there exist at least one positive in the test (equal to $0$ if no positive in the test, and $1$ otherwise), while in the latter case, the result of a test specifies the exact number of positives in that test. Motivated by applications in genome sequence processing, we propose a novel non-adaptive test model termed semi-quantitative group testing (SQGT). This model accounts for the fact that in most applications a test is not precise enough to exactly determine the number of positives, but it is more informative than a simple indicator of the presence of at least one positive. In other words, schemes in which results are obtained using a test device with \emph{limited} precision may be modeled as instances of SQGT\footnote{One may view the SQGT scheme as a generalization of thresholded group testing to multiple thresholds and zero gaps~\cite{D06}.}. 

We also allow for the possibility of having different amounts of sample material for different test subjects, which results in non-binary test matrices. Although binary testing is required for some applications -- such as coin weighing -- in other applications, such as conflict resolution in multiple access channel (MAC) and genotyping, non-binary tests may be used to further reduce the number of tests. In the former example, different non-binary values in a test correspond to different power levels of the users, while in the latter example, they correspond to different amounts of genetic material of different subjects. The reason that non-binary tests are extremely important is that in applications like genotyping, tests are very expensive so that one may be inclined to reduce the number of tests at the expense of extracting more genetic material. While there exists information theoretic analysis that is applicable to the non-binary test matrices~\cite[Ch. 6]{D04}, to the best of the authors' knowledge, the only attempts of non-binary code construction relevant to group testing is limited to a handful of papers, including~\cite{J95} and~\cite{CW99}, where constructions are considered for an \emph{adder} MAC channel.

For the new and versatile model of SQGT with $Q$-ary test results and $q$-ary test sample sizes, $Q,q \geq 2$, we define the concept of capacity. Furthermore, we define a 
new generalization of the family of \emph{disjunct} codes, first introduced in~\cite{KS64}, called ``SQ-disjunct'' codes. Similar to the family of disjunct codes, this new code family 
may be decoded using a simple, low-complexity algorithm. We conclude our exposition with a probabilistic method for code construction, of use in applications where the physics of the experiments prohibits structured codes.

The paper is organized as follows. Section~\ref{sec:model} describes the SQGT model, while Section~\ref{sec:informationtheory} introduces the capacity of SQGT. 
In Section~\ref{sec:disjunct}, we define SQ-disjunct codes and present some simple properties of these codes. In Section~\ref{sec:construction}, we describe a number of constructions for SQGT codes.

\section{Semi-quantitative Group Testing}\label{sec:model}

Throughout this paper, we adopt the following notation. Bold-face upper-case and bold-face lower-case letters denote matrices and vectors, respectively. Calligraphic letters are used to denote sets. Asymptotic symbols such as $\sim$ and $o(\cdot)$ are used in a standard manner. For an integer $k$, we define $[k]:=\{0,1,\cdots,k-1\}$. 

Let $N$ denote the number of test subjects, and let $m$ denote the number of positives. Also, let $u$ denote an upper bound on the number of positives (i.e. $m\leq u$). 
Let $S_i$ denote the $i^{\textnormal{th}}$ subject, $i\in\{1,2,\cdots,N\}$, and let $S_{i_j}=D_j$ be the $j^{\textnormal{th}}$ positive, $j\in\{1,2,\cdots,m\}$. Furthermore, let $\mathcal{D}$ denote the set of positives, so that 
$|\mathcal{D}|=m$. We assign to each subject a unique $q$-ary vector of length $n$, termed the ``signature'' or the codeword of the subject. Each coordinate of the signature corresponds to a test. 
If $\textbf{x}_i\in[q]^n$ denotes the signature of the $i^{\textnormal{th}}$ subject, then the $k^{\textnormal{th}}$ coordinate of $\textbf{x}_i$ may be viewed as the ``amount'' of $S_i$ (i.e. sample size, concentration, etc.) 
used in the $k^{\textnormal{th}}$ test. Note that the symbol $0$ indicates that $S_i$ is not in the test. For convenience, we refer to the collection of codewords arranged column-wise as the \emph{test matrix} or \emph{code}.

The result of each test is an integer from the set $[Q]$. Each test outcome depends on the number of positives and their sample amount in the test through $Q$ thresholds, $\eta_l$ ($l\in\{1,2,\cdots,Q\}$). 
More precisely, the outcome of the $k^{\textnormal{th}}$ test, $y_k$, equals
\vspace{-0.09cm}\begin{equation}\label{model1}
y_k=r\ \ \ \ \ \textnormal{if}\ \ \ \eta_r\leq\sum_{j=1}^mx_{k,i_j}< \eta_{r+1},
\vspace{-0.09cm}\end{equation}
where $x_{k,i_j}$ is the $k^{\textnormal{th}}$ coordinate of $\textbf{x}_{i_j}$, and $\eta_0=0$.
Based on the definition, it is clear that SQGT may be viewed as a concatenation of an adder channel and a decimator (quantizer). 
Also, if $q=Q=2$ and $\eta_1=1$, the SQGT model reduces to conventional GT. Furthermore, if $Q-1=m(q-1)$ and $\forall r\in[Q]$, $\eta_r=r$, then the 
SQGT reduces to the adder channel, with a possibly non-binary test matrix. Note that in this model, we assume that $\eta_Q>(q-1)u$.

Of special interest is SQGT with a uniform quantizer - i.e. SQGT with equidistant thresholds. In this case, $\eta_r=r\eta$, where $r\in[Q+1]$, 
and the following definition may be used to simplify~\eqref{model1}.

\begin{defin}
The ``SQ-sum'' of $s\geq 1$ codewords $\mathbf{x}_j\in [q]^n$, $1\leq j\leq s$, 
denoted by $\mathbf{y}=\bigoasterisk_{j=1}^{s}\mathbf{x}_j=\mathbf{x}_1\oasterisk\mathbf{x}_2\oasterisk\cdots\oasterisk\mathbf{x}_s$, is a 
vector of length $n$ with its $i^{\textnormal{th}}$ coordinate equal to 
\vspace{-0.09cm}\begin{equation}
y_i=\left\lfloor\frac{x_{i,1}+x_{i,2}+\cdots+x_{i,s}}{\eta}\right\rfloor.
\vspace{-0.09cm}\end{equation}
Here, ``$+$'' stands for real-valued addition.
\end{defin}
Using this definition,~\eqref{model1} reduces to
\vspace{-0.09cm}\begin{equation}
\mathbf{y}=\bigoasterisk_{j=1}^{m}\mathbf{x}_{i_j}.
\vspace{-0.09cm}\end{equation}
where $\mathbf{x}_{i_j}$ is the signature of the $j^{\textnormal{th}}$ positive.

\section{Capacity of SQGT}\label{sec:informationtheory}
It is well-known that group testing may be viewed as a special instance of a multiple access channel (MAC) (e.g. see~\cite{D04}). Using this connection, Malyutov~\cite{M80}, D'yachkov~\cite{D04}, and Atia et al.~\cite{AS10} derived information theoretic necessary and sufficient conditions on the required number of tests for conventional GT. It is tedious, yet straightforward, to show that the model described in \cite{M80,D04,AS10} may also 
be used to evaluate SQGT schemes. The main difference in the analysis of GT and SQGT arises due to the different forms of the mutual information used to express the necessary and sufficient conditions. 
We therefore focus on characterizing the mutual informations arising in the SQGT framework. Our notation follows the setup of~\cite{AS10}.

Let the sample amount of each subject in each test be chosen in an i.i.d manner from a $q$-ary alphabet, according to a distribution $P_T$. 
Also, let $\mathcal{D}^{\{i\}}_1$ and $\mathcal{D}^{\{i\}}_2$ be disjoint partitions of the set of positives, $\mathcal{D}$, such that $|\mathcal{D}^{\{i\}}_1|=i$ and $|\mathcal{D}^{\{i\}}_2|=m-i$; we denote by $\mathcal{A}_{\mathcal{D}}^{\{i\}}$ the set of all possible pairs $(\mathcal{D}^{\{i\}}_1,\mathcal{D}^{\{i\}}_2)$. For a single test,  we define $y$ as the test result, and $\textbf{t}_{\mathcal{D}_j}^{\{i\}}$ (where $j=1,2$) as a vector of size $1\times |\mathcal{D}^{\{i\}}_j|$, with its $k^{\textnormal{th}}$ entry equal to the sample amount of the $k^{\textnormal{th}}$ positive of $\mathcal{D}^{\{i\}}_j$ in the test. Fig. 1 shows a choice of $(\mathcal{D}^{\{2\}}_1,\mathcal{D}^{\{2\}}_2)$ and their corresponding vectors $\textbf{t}_{\mathcal{D}_1}^{\{2\}}$ and $\textbf{t}_{\mathcal{D}_2}^{\{2\}}$ for the case where $m=5$ and $q=2$. 

\begin{figure}
\centering
\subfigure[][]{
\centering
\includegraphics[width=0.18\textwidth]{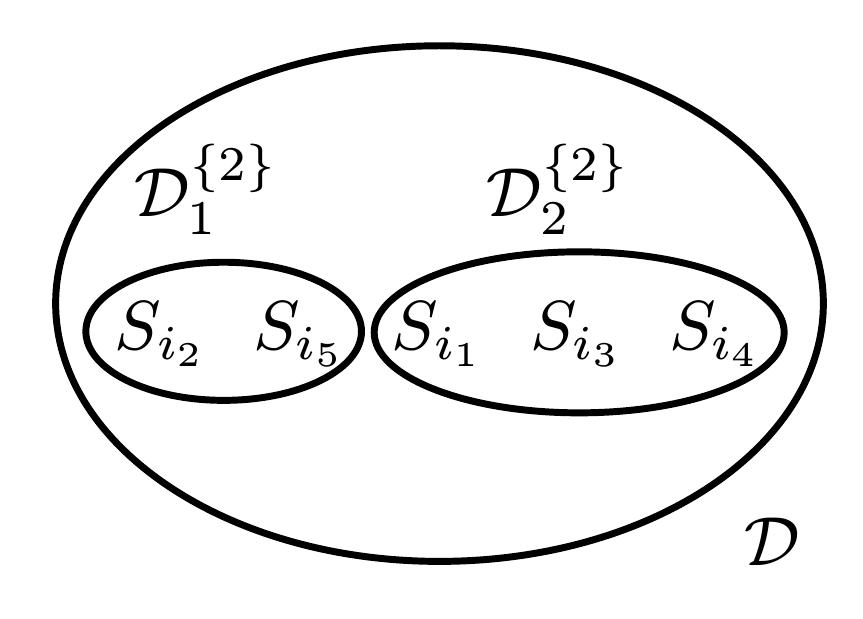}
\label{subfig:partition}}
\subfigure[][]{
\centering
\includegraphics[width=0.15\textwidth]{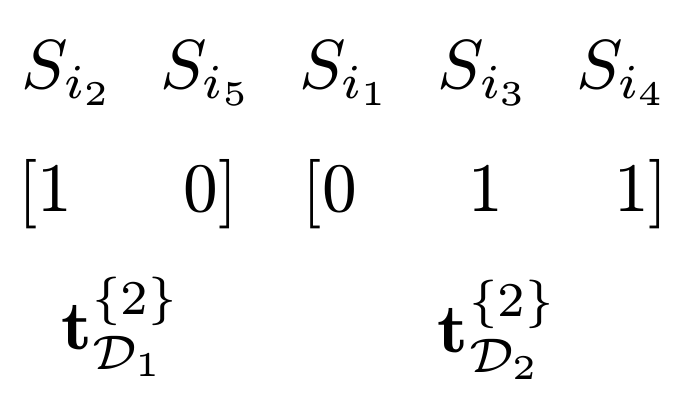}
\label{subfig:partition_vec}}
\vspace{-0.09cm}
\caption{One choice of $(\mathcal{D}^{\{2\}}_1,\mathcal{D}^{\{2\}}_2)$ and their corresponding $\textbf{t}_{\mathcal{D}_1}^{\{2\}}$ and $\textbf{t}_{\mathcal{D}_2}^{\{2\}}$ in a binary test design for $m=5$. }
\label{fig:partition}
\vspace{-0.3cm}
\end{figure}

By following the same steps as in~\cite{AS10}, it can be shown that for any fixed $m$, if the number of tests $n$ satisfies
\vspace{-0.09cm}\begin{equation}\label{sufficient}
n>\max_{i:(\mathcal{D}^{\{i\}}_1,\mathcal{D}^{\{i\}}_2)\in\mathcal{A}_{\mathcal{D}}^{\{i\}}}\frac{\log{{N-m}\choose i}{m\choose i}}{I(\textbf{t}_{\mathcal{D}_1}^{\{i\}};\textbf{t}_{\mathcal{D}_2}^{\{i\}},y)}\ \ \ \ i=1,2,\cdots,m,
\vspace{-0.09cm}\end{equation}
then the average probability of error asymptotically approaches zero\footnote{A sufficient condition for zero error probability can also be found based on the exponential asymptotics of the average error probability. Such calculations are presented in~\cite{D04} for the particular case of conventional GT.  }. In this equation, $I(\textbf{t}_{\mathcal{D}_1}^{\{i\}};\textbf{t}_{\mathcal{D}_2}^{\{i\}},y)$ stands for the mutual information between $\textbf{t}_{\mathcal{D}_1}^{\{i\}}$ and $(\textbf{t}_{\mathcal{D}_2}^{\{i\}},y)$. Note that since the sample amounts of the subjects are chosen independently and identically, the value of $I(\textbf{t}_{\mathcal{D}_1}^{\{i\}};\textbf{t}_{\mathcal{D}_2}^{\{i\}},y)$ does not depend on the specific choice of $(\mathcal{D}^{\{i\}}_1,\mathcal{D}^{\{i\}}_2)$. Similarly, a necessary condition for zero average error probability for SQGT is
\vspace{-0.09cm}\begin{equation}\label{necessary}
n\geq\max_{i:(\mathcal{D}^{\{i\}}_1,\mathcal{D}^{\{i\}}_2)\in\mathcal{A}_{\mathcal{D}}^{\{i\}}}\frac{\log{{N-m+i}\choose i}}{I(\textbf{t}_{\mathcal{D}_1}^{\{i\}};\textbf{t}_{\mathcal{D}_2}^{\{i\}},y)}\ \ \ \ i=1,2,\cdots,m.
\vspace{-0.09cm}\end{equation}

\begin{defin}\label{def:capacity}[Asymptotic capacity of SQGT channel]
Using \eqref{sufficient} and \eqref{necessary}, we define the \emph{asymptotic capacity} of the channel corresponding to the SQGT scheme (henceforth, SQGT channel) as
\vspace{-0.09cm}\begin{equation}\label{eq:capacity}
C=\supr_{P_T,\boldsymbol{\eta}}{\alpha(m,P_T,\boldsymbol{\eta})},
\vspace{-0.09cm}\end{equation}
where $\alpha(m,P_T,\boldsymbol{\eta})=\min_{i=1,2,\cdots,m}\frac{I(\textbf{t}_{\mathcal{D}_1}^{\{i\}};\textbf{t}_{\mathcal{D}_2}^{\{i\}},y)}{i}$, and where $\boldsymbol{\eta}$ is a vector of length $Q$ with $\eta_k$ its $k^{\textnormal{th}}$ entry.
\end{defin}

In certain applications, $\boldsymbol{\eta}$ may be determined a priori by the resolution of the test equipment. In such applications, the only design parameter to optimize is $P_T$.
On the other hand, if one is able to control the thresholds, $\boldsymbol{\eta}$ becomes a design parameter and clearly exhibits a strong influence on the capacity of the test scheme.
 
Define the rate of a group test as $R=\frac{\log N}{n}$. The next theorem clarifies the use of the term ``capacity'' in Definition~\ref{def:capacity}. 
\begin{theorem}
For SQGT, $C=\supr_{P_T,\boldsymbol{\eta}}{I(\textbf{t}_{\mathcal{D}_1}^{\{m\}};\textbf{t}_{\mathcal{D}_2}^{\{m\}},y)}/{m}$, and all rates bellow capacity are achievable. In other words, for every rate $R<C$, there exists a test design for which the average probability of error converges to zero. Conversely, any test design with zero achieving average probability of error must asymptotically satisfy $R<C$.
\end{theorem}
\begin{proof}
Follows from simple modifications of arguments in~\cite{M80},~\cite{D04}. 
\end{proof}


The mutual information $I(\textbf{t}_{\mathcal{D}_1}^{\{m\}};\textbf{t}_{\mathcal{D}_2}^{\{m\}},y)$ in this theorem may be evaluated as follows. Let $W_1$ denote the $l_1$-norm of $\textbf{t}_{\mathcal{D}_1}^{\{m\}}$. Then,
\begin{align}\nonumber
&I_{\textnormal{SQ}}(\textbf{t}_{\mathcal{D}_1}^{\{m\}};\textbf{t}_{\mathcal{D}_2}^{\{m\}},y)=H(y|\textbf{t}_{\mathcal{D}_2}^{\{m\}})-H(y|\textbf{t}_{\mathcal{D}_1}^{\{m\}},\textbf{t}_{\mathcal{D}_2}^{\{m\}})=H(y).
\end{align}
On the other hand $\forall l\in[Q]$,
\begin{align}
P(y = l) = P(\eta_l\leq W_1 <\eta_{l+1} )=\sum_{w_1=\eta_l}^{\eta_{l\!+\!1}-1}\!P_{W_1}(w_1)
\end{align}
where $P_{W_1}(w_1)$ is the probability mass function (PMF) of $W_1$ and can be found using 
\vspace{-0.09cm}\begin{equation}
P_{W_1}(w_1)=P_T(t_1)*P_T(t_2)*\cdots*P_T(t_m),
\vspace{-0.09cm}\end{equation}
where ``$*$'' denotes convolution. 
Note that when $q=2$, 
\begin{align}
P(y = l) = \sum_{j=\eta_l}^{\eta_{l+1}-1}{m\choose j}p^j(1\!-\!p)^{m-j}
\end{align} 
where $p$ is the probability that a subject is present in a test.

Due to the complicated expression for the mutual information for an arbitrary distribution, a closed-form expression for the test capacity cannot be obtained. We therefore evaluated~\eqref{eq:capacity} numerically using a simple search procedure that allows us to quickly determine 
a lower bound on the capacity. Fig. 2 shows the obtained lower bound on the capacity when $q=3$, and $Q=2$ or $Q=3$. Table~\ref{table:capacity} shows one set of probability distributions 
and thresholds achieving this bound for $Q=3$. 

\begin{figure}
\centering
\includegraphics[width=0.3\textwidth]{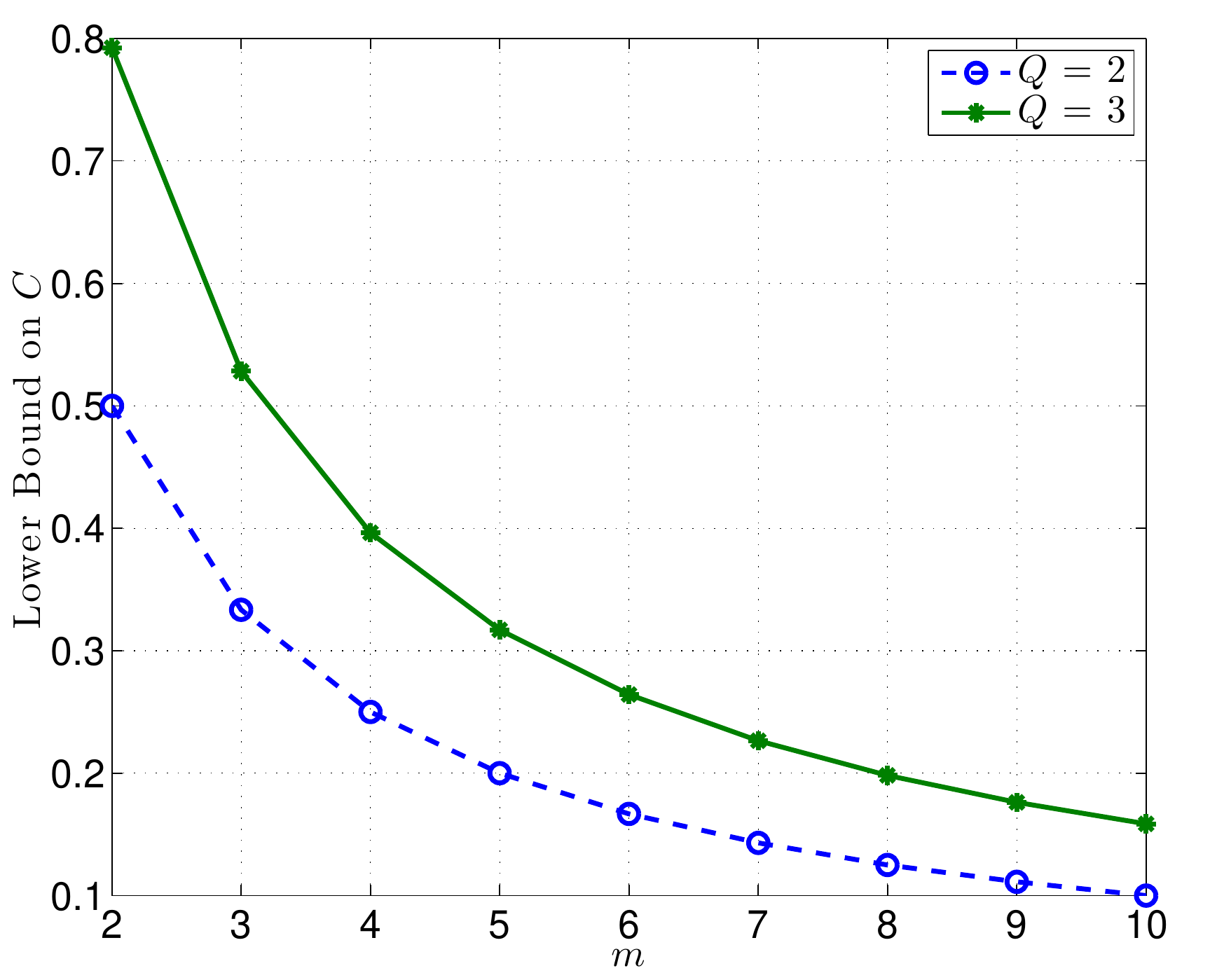}\label{fig:capacity}
\vspace{-0.09cm}
\caption{Numerically obtained lower bound for SQGT with $q=3$ for different values of $m$. }
\vspace{-0.3cm}
\end{figure}

Table~\ref{table:capacity} reveals an interesting property of the quantizers found through numerical search: there exists at least one quantization region that consists of one or two elements only.
What this finding implies is that in order to reduce the number of tests as much as possible, a sufficient amount of qualitative information has to be preserved. For example, by having a quantizer
that assigns the value $v$ only to inputs of value $v$, allows for resolving a large amount of uncertainty. Furthermore, the most informative input, left unaltered after quantization, corresponds
to a statistical average of the input symbols, reminiscent to the centroid of a quantization region. These findings will be discussed in more detail in the full version of the paper.
\begin{table}[t!]
	\centering
	\caption{A set of probability distributions and thresholds corresponding to $Q=3$ in Fig. 2.}
		\begin{tabular}{|c||c|c|}
			\hline 
			$m$ &  $P_T$ & quantizer\\ 
			\hline\hline
			
			$2$ & $[0.33\  0.34 \  0.33]$ & $\{0,1\}\{2\}\{3,4\}$\\
			\hline
			
			$3$ & $[0.43\  0.46\  0.11]$ & $\{0,1\}\{2\}\{3,4,5,6\}$\\			
			
			\hline
			
			$4$ & $[0.18\  0.64\  0.18]$ & $\{0,1,2,3\}\{4\}\{5,6,7,8\}$ \\
			
			\hline
			
			$5$ & $[0.15\  0.70\  0.15]$ & $\{0,1,2,3,4\}\{5\}\{6,7,8,9,10\}$\\
			
			\hline
			
			$6$ & $[0.46\  0.15\  0.39]$&  $\{0,1,2,3,4\}\{5,6\}\{7,8,\cdots,12\}$\\
			
			\hline
			
			$7$ & $[0.34\  0.25\  0.41]$&  $\{0,1,\cdots,6\}\{7,8\}\{9,10,\cdots,14\}$\\	
			
			\hline
			
			$8$ & $[0.10\  0.80\  0.10]$&  $\{0,1,\cdots,7\}\{8\}\{9,10,\cdots,16\}$ \\	
			
			\hline
			
			$9$ & $[0.09\  0.82\  0.09]$ & $\{0,1,\cdots,8\}\{9\}\{10,11,\cdots,18\}$ \\							
			\hline
			
			$10$ & $[0.58\  0.28\  0.14]$ & $\{0,1,\cdots,4\}\{5,6\}\{7,8,\cdots,20\}$ \\	
			\hline					

			\end{tabular}\label{table:capacity}
			\vspace*{-10pt}	
\end{table}

\section{Generalized Disjunct and Separable Codes for SQGT}\label{sec:disjunct}
Disjunct codes were first introduced in~\cite{KS64} for efficient zero-error group testing reconstruction. In what follows, we define a new family of disjunct codes suitable for SQGT that shares many of the properties of binary disjunct codes. 

\begin{defin}
The \emph{syndrome} of a set of vectors $\{\mathbf{x}_i\}$, $i\in\{1,2,\cdots,s\}$, such that $\mathbf{x}_i \in[q]^n$, is a vector $\mathbf{y} \in [Q]^n$ equal to $\mathbf{y}=\bigoasterisk_{j=1}^s\mathbf{x}_j$.
\end{defin}

\begin{defin}
A set of codewords $\mathcal{X}=\{\mathbf{x}_1,\mathbf{x}_2,\cdots,\mathbf{x}_s\}$ with syndrome $\mathbf{y}_{\!_\mathcal{X}}$ is said to be \emph{included} in another set of codewords $\mathcal{Z}=\{\mathbf{z}_1,\mathbf{z}_2,\cdots,\mathbf{z}_t\}$ with syndrome $\mathbf{y}_{\!_\mathcal{Z}}$, if $\forall i\in\{1,2,\cdots,n\}$, ${{y}_{\!_\mathcal{X}}}_{\!_i}\leq {{y}_{\!_\mathcal{Z}}}_{\!_i}$. We denote this inclusion property by $\mathcal{X}\lhd\mathcal{Z}$ or equivalently $\mathbf{y}_{\!_\mathcal{X}}\lhd\mathbf{y}_{\!_\mathcal{Z}}$.  
\end{defin}
\begin{remark}
By this definition, it can be easily verified that if $\mathcal{X}\subseteq\mathcal{Z}$, then $\mathcal{X}\lhd \mathcal{Z}$.
\end{remark}

Note that for $q=2$, this definition is equivalent to the definition of inclusion for conventional GT, defined in~\cite{KS64}.

\begin{defin}
A code is called a $[q;Q;\boldsymbol{\eta};u]$-SQ-disjunct code of length $n$ and size $N$ if $\forall s,t\leq u$  and for any sets of $q$-ary codewords $\mathcal{X}=\{\mathbf{x}_1,\mathbf{x}_2,\cdots,\mathbf{x}_s\}$ and $\mathcal{Z}=\{\mathbf{z}_1,\mathbf{z}_2,\cdots,\mathbf{z}_t\}$, $\mathcal{X}\lhd \mathcal{Z}$ implies $\mathcal{X}\subseteq \mathcal{Z}$.
\end{defin}
Henceforth, we focus on the case where the thresholds are equidistant. We call such codes $[q;Q;\eta;u]$-SQ-disjunct codes.

\begin{prop}\label{SQDisj}
 A code is $[q;Q;\eta;u]$-SQ-disjunct if and only if no codeword is included in the set of $u$ other codewords. 
 \end{prop}
\begin{proof}
It is easy to verify that if a code is $[q;Q;\eta;u]$-SQ-disjunct, then no codeword is included in the set of $u$ other codewords. 

Conversely, let $\mathcal{X}=\{\mathbf{x}_1,\mathbf{x}_2, \cdots,\mathbf{x}_s\}$ and $\mathcal{Z}=\{\mathbf{z}_1, \mathbf{z}_2 , \cdots,\mathbf{z}_t\}$ be two sets of codewords where $s,t\leq u$. From the assumption that no codeword is included in the set of $u$ other codewords, one can conclude that no codeword is included in the set of $t$ other codewords when $t\leq u$. If $\mathcal{X}\lhd\mathcal{Z}$ but $\mathcal{X}\nsubseteq\mathcal{Z}$, then there exists a codeword $\mathbf{x}_j$ ($j\in\{1,2,\cdots,s\}$) such that $\{\mathbf{x}_j\}\nsubseteq\mathcal{Z}$. But since $\{\mathbf{x}_j\}\lhd \mathcal{X}\lhd\mathcal{Z}$, then $\{\mathbf{x}_j\}\lhd\mathcal{Z}$, which contradicts the assumption that no codeword is included in $t$ other codewords.
\end{proof}

\begin{remark}
From Proposition~\ref{SQDisj}, one can conclude that a code is $[q;Q;\eta;u]$-SQ-disjunct if and only if for any set of $u+1$ codewords, $\{\mathbf{x}_1,\mathbf{x}_2,\cdots,\mathbf{x}_{u+1}\}$, there exists a unique coordinate ${k(i)}$ in each codeword $\mathbf{x}_i$, for which 
\vspace{-0.09cm}\begin{equation}\label{condition}
\left\lfloor\frac{x_{k(i),i}}{\eta}\right\rfloor>\left\lfloor\frac{\sum_{j=1,j\neq i}^{u+1}{x}_{k(i),j}}{\eta}\right\rfloor.
\vspace{-0.09cm}\end{equation}
By unique coordinate, we mean that $k(i)\neq k(j)$, if $i\neq j$. Consequently, a necessary condition for the existence of a $[q;Q;\eta;u]$-SQ-disjunct code is that $q-1\geq\eta$. 
As a result, there exist no binary $[2;Q;\eta;u]$-SQ-disjunct code when $\eta>1$.
\end{remark}

\begin{prop}\label{prop:disj}
Any code generated by multiplying a conventional binary $u$-disjunct code by $q-1$, where $q-1\geq\eta_1$, is a $[q;Q;\boldsymbol{\eta};u]$-SQ-disjunct code. 
As a result, the rate of the best $[q;Q;\boldsymbol{\eta};u]$-SQ-disjunct code is at least as large as the rate of the best binary $u$-disjunct code with the same size and length.
\end{prop}

Our interest in SQ-disjunct codes lies in their simple decoding procedures, of complexity $O(nN)$. However, one can construct codes for SQGT using other GT codes, such as binary $u$-separable codes for conventional 
GT~\cite{KS64} or codes designed for the adder channel~\cite{L75}. It can be shown that any of these two family of codes can be multiplied by $\eta$ to form a code for SQGT. 

\begin{remark}
The constructions described in this section reveal the following, and highly intuitive fact: the number of individuals that may be successfully examined with $Q$-ary SQGT may be as large as the number of individuals that may be tested under the adder channel model, provided that one is allowed to pool different amounts of sample material in each test. In other words, the rate of adder and SGQT channels may be the same, despite the loss of information induced by the quantizer, provided that the alphabet size of the latter scheme is sufficiently larger than the alphabet size of the former scheme.
\end{remark}

\section{Code Construction for SQGT}\label{sec:construction}

In what follows, we discuss two approaches for constructing SQGT codes. For simplicity, we focus on SQGT codes with equidistant thresholds. The first approach relies on classical combinatorial methods, while the second approach relies on probabilistic methods. The second approach is of special interest for applications such as genotyping, where one cannot arbitrarily choose the test matrices. The tests are usually determined by the physics of the experiment, and only certain statistical properties of the tests are known. In this scenario, ``structure'' is to be seen as probabilistic trait. We show that one way to approach this problem is to characterize the number of tests that ensures that \emph{almost all members of a code family} possess a given trait and act as SQGT codes.

\subsection{Combinatorial Construction}
 Fix a binary $u$-disjunct code matrix $\mathbf{C}_b$ of dimensions $n_b\times N_b$, with code-length $n_b$ and $N_b$ codewords. Let $K=\left\lfloor\log_u\left(\left(\frac{q-1}{\eta}\right)(u-1)+1\right)\right\rfloor$; construct a code of length $n=n_b$ and size $N=KN_b$ by concatenating $K$ matrices, $\mathbf{C}=[\mathbf{C}_1,\mathbf{C}_2,\cdots, \mathbf{C}_K]$, where $\mathbf{C}_j=\left(\sum_{i=0}^{j-1}u^i\eta\right)\mathbf{C}_b,$ $1\leq j\leq K$. 

\begin{theorem}
Let the concatenated code $\mathbf{C}$ be as described above. The code is capable of uniquely identifying up to $u$ positives.  
\end{theorem}
\begin{proof}
The proof is based on exhibiting a decoding procedure and showing that the procedure allows for distinguishing between any two different sets of positives. The decoder is described below.

Let $\mathbf{y}$ be the $Q$-ary vector of test outcomes, or equivalently, the syndrome of the positives. For a rational vector $\mathbf{z}$, let $\left\lfloor\mathbf{z}\right\rfloor$ and $\left\langle\mathbf{z}\right\rangle$ denote the vector of integer parts of $\mathbf{z}$ and fractional parts of $\mathbf{z}$, respectively. If $u=1$, decoding reduces to finding the column of $\mathbf{C}$ equal to $\eta \mathbf{y}$. 
 If $u>1$, decoding proceeds as follows.

\textbf{Step 1:} Set $\mathbf{y}'_{K}=\mathbf{y}$ and form vectors $\mathbf{y}_j$, $1\leq j\leq K$, using the rules:
\vspace{-0.09cm}\begin{equation}
\mathbf{y}_j=\left(\frac{u^j-1}{u-1}\right)\left\lfloor\left(\frac{u-1}{u^j-1}\right)\mathbf{y}'_{j}\right\rfloor,
\vspace{-0.09cm}\end{equation}
and
\vspace{-0.09cm}\begin{equation}
\mathbf{y}'_{j-1}=\left(\frac{u^j-1}{u-1}\right)\left\langle\left(\frac{u-1}{u^j-1}\right)\mathbf{y}'_{j}\right\rangle.
\vspace{-0.09cm}\end{equation}

\textbf{Step 2:} Identify the positives as follows: if the syndrome of a column of $\mathbf{C}_j$ is included in $\mathbf{y}_j$, declare the subject corresponding to that column positive. Declare the
subject negative otherwise. 

The result is obviously true for $u=1$. Therefore, we focus on the case $u>1$. First, using induction, one can prove that each $\mathbf{y}_j$, $1\leq j\leq K$, is the syndrome of a subset of columns of $\mathbf{C}_j$ corresponding to positives. Let $\mathbf{C}'_j=[\mathbf{C}_1,\mathbf{C}_2,\cdots, \mathbf{C}_j]$, where $1\leq j\leq K$. Since the non-zero entries of $\mathbf{C}$ are multiples of $\eta$, $\eta\mathbf{y}$ is the sum of columns of $\mathbf{C}$ corresponding to a subset of positives. Also, the maximum value of the entries of $\mathbf{C}'_{K-1}$ equals $\eta\frac{u^{K-1}-1}{u-1}$. Since there are at most $u$ positives, the maximum value of their sum does not exceed $\eta\frac{u^K-u}{u-1}$. This bound is strictly smaller than $\eta\frac{u^K-1}{u-1}$, the minimum non-zero entry of $\mathbf{C}_K$. As a result, $\mathbf{y}_K$ is the syndrome of the positives with signatures in $\mathbf{C}_K$, and $\mathbf{y}'_{K-1}$ is the syndrome of positives with signatures in $\mathbf{C}'_{K-1}$. Similarly, it can be shown that $\forall j, 1\leq j\leq K-1$, $\mathbf{y}_j$ is the syndrome of the positives with signature in $\mathbf{C}_j$, 
and $\mathbf{y}'_{j-1}$ is the syndrome of the positives with signatures in $\mathbf{C}'_{j-1}$.

From Proposition~\ref{prop:disj}, we know that each $\mathbf{C}_j$ is a $[q;Q;\eta;u]$-SQ-disjunct code. Consequently using step 2,  one can uniquely identify the positives with signatures from $\mathbf{C}_j$.
\end{proof}
%

 \begin{remark}
The method described above can be used with any binary separable code for conventional GT or adder channel to generate a SQGT code.
 \end{remark}
 
 \begin{remark}
All the constructions described in this paper are able to identify up to $u$ positives in a  pool of $N$ subjects when $u\ll N$. However, when $0\leq u\leq N$, one can construct non-binary codes with length $n$ and asymptotic size of $N\sim(\lfloor\log\lfloor\frac{q-1}{\eta}\rfloor\rfloor+\log n/2)n$ (see the full version of this paper).
\end{remark}
 
\subsection{Probabilistic Construction}
We consider the following problem: find a critical rate such that \emph{any} randomly generated $q$-ary code with rate less than the critical rate is a $[q;Q;\eta;u]$-SQ-disjunct code \emph{with probability close to one}. 
Based on the critical rate, which depends on the statistical properties of the process used to generate the codes, one can identify the smallest number of tests required to ensure that any code in the family
may be used for SGGT.

\begin{theorem}
Let $R_{\textnormal{critical}}=\frac{\log \gamma}{u+1}+\frac{\log(\epsilon u!)}{n(u+1)}$, $\forall\epsilon>0$, where $\gamma=q^{(u+1)}/A$, $A=\eta{(I-1)\eta+u\choose u+1}-\eta{\eta+u\choose u+1}+(q-I\eta){I\eta+u-1\choose u}$, and $I=\lfloor \frac{q-1}{\eta}\rfloor$. Any $q$-ary code of length $n$ and size $N$ with rate asymptotically satisfying $R\leq R_{\textnormal{critical}}$ is a $[q;Q;\eta;u]$-SQ-disjunct code with probability at least $1-\epsilon$.
\end{theorem}

\begin{proof}
Let $\mathcal{C}$ be a code of length $n$ and size $N$, and let $\mathcal{M}$ be a set of $u+1$ codewords of $\mathcal{C}$. There are $L=\binom{N}{u+1}$ different ways to choose $\mathcal{M}$. For the $i^{\textnormal{th}}$ choice of $\mathcal{M}$, we define $E_{i}$ as the event that the syndrome of at least one of the codewords in $\mathcal{M}$ is included in the syndrome of the other $u$ codewords. Suppose that $P(E_i)\leq p'$ for all $i$; using the union bound, $P\left(\bigcup_{i=1}^{L}E_i\right)\leq Lp'$. Therefore, if 
\vspace{-0.09cm}\begin{equation}\label{lovasz-high}
p'\leq\frac{\epsilon}{L},
\vspace{-0.09cm}\end{equation}
then $P\left(\bigcap_{i=1}^{L}\bar{E}_{i}\right)\geq1-\epsilon$, where $\bar{E}_{i}$ is the complement of the event $E_i$. In other words, $\mathcal{C}$ is a  $[q;Q;\eta;u]$-SQ-disjunct code with probability at least $1-\epsilon$.

From the definition of $E_i$, one has $P(E_i)=\frac{a}{(u+1)!{q^n\choose u+1}}$, where $a$ is the number of $q$-ary matrices of size $n\times(u+1)$ that do not satisfy~\eqref{condition} and have distinct columns; also, $(u+1)!{q^n\choose u+1}$ is the total number of $n\times(u+1)$ matrices with distinct columns. In order to find an upper bound on $a$, we use the fact that a matrix that satisfies~\eqref{condition} has distinct columns. Consequently, 
\vspace{-0.09cm}\begin{equation}
(u+1)!{q^n\choose u+1}-a=q^{n(u+1)}-b
\vspace{-0.09cm}\end{equation}
where $q^{n(u+1)}$ is the number of $q$-ary matrices with (possibly) repeated columns, and $b$ is the number of such matrices that satisfy~\eqref{condition}. It can be easily seen that, 
\vspace{-0.09cm}\begin{equation}
b\leq (u+1)c
\vspace{-0.09cm}\end{equation}
where $c$ is the number of $q$-ary matrices that do not contain a row, $\mathbf{x}$, satisfying $\lfloor\frac{x_{1}}{\eta}\rfloor>\lfloor\frac{\sum_{j=2}^{u+1}{x}_{j}}{\eta}\rfloor$. On the other hand, $c=A^n$ where $A$ is the number of ``acceptable'' $q$-ary rows of length $u+1$. Let $\mathbf{x}\in[q]^{u+1}$ denote an acceptable row. If $A_i$, $i\in\{1,2,\cdots,\lfloor \frac{q-1}{\eta} \rfloor\}$, denotes the number of acceptable rows with the first entry $x_1$ from $\{i\eta,i\eta+1,\cdots,(i+1)\eta-1\}$, then $A=\sum_iA_i$. Let $I=\lfloor \frac{q-1}{\eta}\rfloor$. If $i<I$, there are $\eta$ choices for $x_1$; if $i=I$, we have $(q-I\eta)$ choices for $x_1$. The number of ways to choose the rest of the entries (denoted by $B_i$) is
\vspace{-0.09cm}\begin{equation}
B_i=\sum_{k=0}^{i\eta-1}{k+u-1\choose u-1}={i\eta+u-1\choose u}
\vspace{-0.09cm}\end{equation}
where ${k+u-1\choose u-1}$ counts the number of non-negative integer solutions to $\sum_{j=2}^{u+1}{x_j}=k$. Consequently,
\begin{align}
A=&\  \eta\sum_{i=1}^{I-1}\!{i\eta\!+\!u\!-\!1\choose u}\!+\!(q\!-\!I\eta){I\eta\!+\!u\!-\!1\choose u}\\\nonumber
=&\  \eta{\!(I\!-\!1)\eta\!+\!u\!\choose u+1}\!-\!\eta{\eta\!+\!u\choose u\!+\!1}\!+\!(q\!-\!I\eta){\!I\eta\!+\!u\!-\!1\!\choose u}.
\end{align}
Using these results,
\vspace{-0.09cm}\begin{equation}\label{myp}
P(E_i)\leq 1-\frac{q^{n(u+1)}-(u+1)A^n}{(u+1)!{q^n\choose u+1}}\leq \frac{A^n}{u!{q^n\choose u+1}}.
\vspace{-0.09cm}\end{equation}
Note that the second inequality does not loosen the bound significantly since $1-\frac{q^{n(u+1)}-(u+1)A^n}{(u+1)!{q^n\choose u+1}}\sim  \frac{A^n}{u!{q^n\choose u+1}}$ as $n\rightarrow\infty$. 

As $n,N\rightarrow\infty$, $p'\sim \frac{(u+1)A^n}{q^{n(u+1)}}$ and $L\sim\frac{N^{u+1}}{(u+1)!}$.  Consequently,~\eqref{lovasz-high} asymptotically simplifies to 
\begin{align}\nonumber
N^{u+1}\leq \gamma^n\epsilon u!
&\Rightarrow R\leq R_{\textnormal{critical}}=\frac{\log \gamma}{u+1}+\frac{\log(\epsilon u!)}{n(u+1)}
\end{align}
where $\gamma=q^{(u+1)}/A$. This completes the proof.
\end{proof}
\vspace{10pt}
\textbf{Acknowledgments:} This work was supported by the NSF grants  CCF 0821910, CCF 0809895, and CCF 0939370, and an NSERC postgraduate scholarship.


\end{document}